\documentclass[final,5p,times,twocolumn,nopreprintline]{elsarticle}

\usepackage[T1]{fontenc}
\usepackage{amsmath,amssymb,amsfonts}
\usepackage{graphicx}
\usepackage{physics}
\usepackage{hyperref}
\hypersetup{hidelinks}
\usepackage{geometry}
\usepackage{booktabs}

\renewcommand{\emph}[1]{#1}
\newcommand{\Event}{\mathcal{E}}
\newcommand{\History}{\mathcal{H}}
\newcommand{\Bandwidth}{\mathfrak{B}}
\newtheorem{prop}{Proposition}
\newtheorem{defn}{Definition}
\newenvironment{proof}{\par\noindent\textit{Proof.}\ }{\hfill$\square$\par\medskip}
\journal{Annals of Physics}

\begin{document}

\begin{frontmatter}

% ----------------------------------------------------------------------
% Title
% ----------------------------------------------------------------------
\title{Events as Spacetime Anchors:\\
Local Irreversibility at the Interface of Quantum Field Theory and Relativity\tnoteref{accepted}}
\tnotetext[accepted]{\raggedright Accepted for publication in \emph{Annals of Physics}. \copyright \ 2026. This manuscript version is made available under
the \href{https://creativecommons.org/licenses/by-nc-nd/4.0/}{CC BY-NC-ND 4.0 license}.}

% ----------------------------------------------------------------------
% Author
% ----------------------------------------------------------------------
\author{Shuo Zhang}
\address{Department of Astronomy, Tsinghua University, Beijing 100084, China}
\ead{szhang.astro@gmail.com}

% ----------------------------------------------------------------------
% Abstract
% ----------------------------------------------------------------------
\begin{abstract}
The reconciliation of quantum field theory (QFT) with general relativity (GR) is often framed as a problem of unifying dynamical laws. In this work, we argue that an underappreciated aspect of this tension concerns the descriptive primitives of the two theories: GR is naturally organized around spacetime events and their causal ordering, while QFT is organized around states, operators, and unitary evolution, without an intrinsic criterion for when a quantum process culminates in a definite spacetime fact.

We propose an event-centered framework in which events, defined as locally irreversible records generated by quantum--environment interactions, serve as the minimal interface between quantum dynamics and relativistic spacetime structure. To model the transition from quantum potentiality to relativistic facticity, we introduce a set of operational criteria based on quantum information theory: local classicalization, redundant record formation, and irreversibility against recovery. This transition, termed local generative freezing, is illustrated through an explicit repeated-collision toy model. The model serves as an illustrative open-system realization of the criteria rather than a derivation from relativistic QFT. Within globally unitary dynamics, the freezing time $t_*$ can be obtained in closed form. The model exhibits a threshold-independent certification structure: local classicalization is certified no later than channel-level irrecoverability for any admissible tolerances, and earlier for generic initial states up to the integer rounding inherent in the collision model. The completion of anchoring is determined by the more demanding of the two requirements, irrecoverability and redundancy, thereby organizing the parameter space into distinct anchoring regimes.

In this picture, physical history is organized, at the level of anchored records, not as a continuous trajectory in state space but as a partially ordered skeleton of frozen events, upon which effective field-theoretic descriptions operate. The proposed framework addresses the question of event anchoring, namely when a candidate outcome becomes a stable spacetime fact, while remaining fully compatible with standard QFT and relativistic causality. It does not derive the Born rule or resolve single-outcome selection, which pertain to the separate question of event generation.
\end{abstract}


\end{frontmatter}

% ======================================================================
\section{Introduction: The Ontological Status of Events}
\label{sec:intro}
% ======================================================================

The reconciliation of quantum field theory (QFT) with general relativity (GR) is often framed as a technical problem involving quantization schemes or ultraviolet completion. While such issues are important, they obscure an often-overlooked aspect of the tension between the two theories: they employ different descriptive primitives.

General relativity is naturally organized around the notion of events. Physical reality is described as a spacetime history composed of localized occurrences and their causal relations. Although coordinates and time parameters depend on the observer, the causal ordering of events is invariant and constitutes the physical content of the theory. In this sense, spacetime is not merely a background but a structured record of definite facts.

Quantum field theory, by contrast, does not treat events as primitive elements of its formalism. Its fundamental description involves quantum states evolving unitarily, field operators, and correlation functions \cite{Haag1996}. Interactions are described as continuous and, in principle, reversible processes in Hilbert space. While QFT successfully predicts transition amplitudes and correlation patterns, it does not by itself provide an intrinsic criterion for when a physical process culminates in a definite outcome that can be regarded as a spacetime fact.

To be clear, QFT is not devoid of localized happenings. Operator insertions, interaction vertices, and detector responses all describe localized processes. However, these elements either remain part of a globally coherent amplitude (as in perturbative scattering theory) or acquire definiteness only through external measurement postulates (as in detector models). Standard QFT does not provide an internal criterion for when such a process becomes an irreversible, redundantly recorded spacetime fact. This is the gap that the present work aims to address.

This question has received renewed attention in recent years. Work on quantum reference frames has shown that events and their localization are relative to the operational capabilities of a laboratory or observer \cite{Vilasini2025,Apadula2024}. The Fewster--Verch measurement framework has demonstrated how measurement events can be constructed as local dynamical couplings between target and probe fields, avoiding the pathologies of instantaneous projection in relativistic settings \cite{Mandrysch2024}. Studies of temporal quantum reference frames have revealed that event ordering itself can become indefinite when clocks have finite energy resources \cite{Hausmann2023}. These developments point toward a broadly operational and relational understanding of events in quantum theory. The present work contributes to this broader program by proposing specific information-theoretic criteria for when a quantum process becomes an anchored spacetime fact.

A central distinction, which we make explicit throughout the manuscript, is between event generation and event anchoring. Event generation concerns how a particular outcome is first produced or selected from among alternatives---a question intimately connected to the measurement problem, the Born rule, and single-outcome selection. Event anchoring concerns a subsequent question: once a candidate outcome exists, when does it become sufficiently classicalized, redundantly recorded, and effectively irrecoverable to function as a stable element of spacetime history? The present paper addresses anchoring only. We regard this as a necessary component of any eventual account of event generation, but we do not claim that it constitutes such an account on its own.

The paper is organized as follows. Section~\ref{sec:events_rel_qft} analyzes the different roles of events in relativity and in QFT. Section~\ref{sec:decoherence} clarifies the relation between the present framework and decoherence. Section~\ref{sec:freezing} introduces the notion of local generative freezing and provides operational criteria for event anchoring. Section~\ref{sec:bandwidth} discusses the role of finite information-processing capacity. Section~\ref{sec:toymodel} presents an explicit repeated-collision toy model that realizes the freezing conditions analytically within globally unitary dynamics. Section~\ref{sec:constructing_history} develops the construction of spacetime history as a partially ordered set of events. Section~\ref{sec:discussion} discusses conceptual implications and limitations. Finally, Section~\ref{sec:conclusion} presents a conclusion. 

% ======================================================================
\section{Events in Relativity and in Quantum Field Theory} % REVISED TITLE
\label{sec:events_rel_qft}
% ======================================================================

In relativistic physics, an event corresponds to a localized spacetime occurrence whose physical significance is determined by its causal relations rather than by its coordinate description. Both special and general relativity describe physical reality in terms of a manifold equipped with a causal structure \cite{Einstein1915,Minkowski1908}, where the primary content is encoded in relations such as timelike, spacelike, or null separation between events. Events and their causal ordering are invariant under changes of reference frame and therefore constitute the minimal building blocks of relativistic spacetime.

Quantum field theory contains localized structures---operator insertions, interaction vertices, and detector couplings---but does not by itself provide a criterion for when such structures become definite, irreversible spacetime facts. Operator insertions encode coherent amplitudes rather than definite occurrences, while detector clicks become definite only within an additional measurement framework. Its basic objects are quantum states and field operators evolving unitarily \cite{Haag1996}. Even when interactions occur, QFT describes them as continuous dynamical processes rather than discrete spacetime facts. The theory provides amplitudes for transitions and correlations, but it does not, by itself, single out when or where a particular outcome becomes an objective element of spacetime history.

This contrast is especially evident in the treatment of measurement. In relativistic spacetime, a measurement outcome is naturally treated as an event: it occurs at a definite spacetime location and can influence future events within its future light cone. In QFT, however, measurement outcomes are not primitive elements of the theory; they are introduced through additional interpretational or operational postulates. The formalism itself contains no intrinsic criterion for identifying the completion of an interaction as a spacetime event.

To sharpen this distinction, it is useful to separate three layers of structure:
\begin{enumerate}
    \item Localized interaction: a quantum process occurring within a bounded spacetime region. QFT provides these.
    \item Candidate record: a retainable imprint of the interaction outcome, formed across a system--environment boundary. Decoherence theory addresses these.
    \item Anchored event: the candidate record has become classicalized, redundantly broadcast, and effectively irrecoverable---functioning as a stable node in a causal network. This is what the present framework formalizes.
\end{enumerate}

The gap addressed in this paper is the transition from layer (2) to layer (3). Bridging this gap requires identifying operational conditions under which a quantum process produces an outcome that is not merely locally classical in appearance, but irreversibly anchored as a spacetime fact. Crucially, this must be achieved within globally unitary dynamics, without modifying the fundamental quantum formalism. This compatibility requirement is nontrivial: holographic dualities such as AdS/CFT demonstrate that globally unitary boundary dynamics can coexist with an effective semiclassical spacetime description in the bulk, but the duality itself does not specify when a particular bulk configuration becomes a definite classical event. The present framework addresses precisely this specification.

% ======================================================================
\section{Relation to Decoherence and the Status of Events}
\label{sec:decoherence}
% ======================================================================

It is important to clarify the precise relation between the present framework and the standard theory of decoherence \cite{Zurek2003,Schlosshauer2007}, as the two address related but distinct questions.

Decoherence provides a dynamical explanation for the suppression of quantum interference in the reduced state of a system due to entanglement with environmental degrees of freedom. Given a system $S$ interacting with an environment $E$, decoherence describes the emergence of an approximately diagonal reduced state in a preferred (pointer) basis:
\begin{equation}
\rho_S(t) \;\longrightarrow\; \Delta(\rho_S(t)),
\end{equation}
such that $\|\rho_S(t) - \Delta(\rho_S(t))\|_1 \le \varepsilon_d$ for some small $\varepsilon_d$. This explains why interference effects become inaccessible for local observables on $S$.

However, decoherence alone does not specify when, or whether, a particular outcome has become an objective element of spacetime history. The reduced density matrix after decoherence represents an effective mixture, and its diagonalization does not by itself establish an irreversible, causally embedded record. From a global perspective, the joint system--environment state remains pure and evolves unitarily; the apparent classicality is a feature of the reduced description.

The present framework addresses a different question: how does a specific outcome become anchored as an irreversible element of spacetime history that can participate in relativistic causal relations? This is not a question about the structure of a density matrix, but about the formation of objective records that define the causal skeleton of spacetime.

The two frameworks are therefore complementary rather than competing:
\begin{itemize}
    \item Decoherence concerns the state-level disappearance of interference.
    \item Local generative freezing concerns the channel-level emergence of irreversibly recorded outcomes that can be embedded into a relativistic causal network.
\end{itemize}
In short, decoherence explains why superpositions are not observed, whereas freezing explains when an outcome becomes a spacetime event. As we demonstrate in Section~\ref{sec:toymodel}, these two transitions are genuinely distinct: in the explicit toy model, classicalization is certified no later than channel-level irrecoverability for any admissible tolerances ($\varepsilon_{eb} \le \varepsilon_d$)---and strictly earlier for generic initial states---while full anchoring further awaits redundant record formation.

% ======================================================================
\section{Local Generative Freezing: Definition and Criterion}
\label{sec:freezing}
% ======================================================================

We now introduce the notion of local generative freezing, which characterizes the anchoring of events from quantum interactions in operational and information-theoretic terms.

\subsection{Definition of an Anchored Event} % REVISED TITLE

\begin{defn}[Anchored Event]
\label{def:anchored}
An anchored event is a spacetime-localized quantum--environment interaction whose outcome is classicalized, redundantly recorded in the environment, and effectively irrecoverable, so that the outcome can be embedded as a node in a relativistic causal structure.
\end{defn}

This definition emphasizes three essential features: (i) a physical interaction, (ii) irreversible record formation, and (iii) causal embeddability. To make this notion operational, we introduce a set of jointly sufficient conditions for event anchoring.

\subsection{Freezing Conditions}

Consider a system $S$ interacting with an environment $E$, which may be decomposed into approximately independent fragments $E_1, E_2, \dots, E_N$. Let $X$ denote the classical outcome variable associated with a pointer observable of $S$, and let $\Delta$ denote the completely dephasing map in the pointer basis.

Local generative freezing is said to occur at the earliest time $t_*$ such that the following three conditions are simultaneously satisfied:

\paragraph{(A) Local Classicalization}
The reduced state of the system is approximately diagonal in the pointer basis:
\begin{equation}
\|\rho_S(t) - \Delta(\rho_S(t))\|_1 \le \varepsilon_d .
\end{equation}
This condition corresponds to effective decoherence but is not sufficient on its own.

\paragraph{(B) Redundant Record Formation}
Information about the outcome $X$ is redundantly imprinted in multiple environmental fragments \cite{Zurek2009,Blume-Kohout2006,Riedel2012}:
\begin{equation}
I(X:E_i) \ge (1-\delta)\, H(X)
\quad \text{for many } i ,
\end{equation}
where $I$ denotes mutual information and $H(X)$ is the Shannon entropy of $X$. Define the redundancy $R_\delta$ as the number of disjoint fragments that independently carry $(1-\delta)H(X)$ bits of information about $X$. Freezing requires
\begin{equation}
R_\delta(t) \ge R_* ,
\end{equation}
where $R_*$ is a threshold indicating that the outcome has become publicly accessible and robust. When individual fragments carry insufficient information, groups of $m$ consecutive fragments may be considered; the minimum group size $m_\delta$ and the resulting grouped redundancy $R_\delta = \lfloor n/m_\delta \rfloor$ are computed explicitly in Section~\ref{sec:toymodel}.

\paragraph{(C) Irreversibility Against Recovery}
After tracing out inaccessible environmental degrees of freedom, the effective channel experienced by $S$ must have lost the ability to transmit pointer-basis coherence or entanglement with an external reference. Formally, the induced channel
\begin{equation}
\mathcal{N}_{S \to S}(t) = \mathrm{Tr}_E \circ \mathcal{U}_{SE}(t)
\end{equation}
must act, to within tolerance, as a classical readout in the pointer basis. Let $\mathcal{C}_X$ denote the set of \emph{pointer measure-and-prepare channels}, i.e., channels of the form $\Phi(\rho) = \sum_k \langle k|\rho|k\rangle\,\sigma_k$ with each re-prepared state $\sigma_k$ diagonal in the pointer basis; every element of $\mathcal{C}_X$ is entanglement breaking \cite{Horodecki2003}. The criterion reads
\begin{equation}
\min_{\Phi \in \mathcal{C}_X}
\|\mathcal{N}_{S \to S}(t) - \Phi\|_\diamond
\le \varepsilon_{eb}.
\end{equation}
Since $\mathcal{C}_X \subset \mathrm{EB}$, closeness to $\mathcal{C}_X$ implies closeness to the entanglement-breaking set, so the physical motivation---loss of any capability to transmit quantum correlations---is retained; we adopt the smaller set because it is exactly computable in the model of Section~\ref{sec:toymodel} and because it cleanly implies condition (A), as discussed below. This condition is a statement about the effective reduced channel, not about the global dynamics. The full system--environment state may remain pure and evolve unitarily; what is required is that the channel experienced by the local system act as an approximate pointer-basis readout for \emph{all} input states, including inputs entangled with an external reference. This is a strictly stronger requirement than condition (A), which certifies classicality for one specific state.

\paragraph{Logical Relations Among Conditions}
Condition (C) implies condition (A) quantitatively: if $\mathcal{N}_{S\to S}$ is within $\varepsilon_{eb}$ of some $\Phi \in \mathcal{C}_X$ in diamond norm, then \emph{every} evolved reduced state satisfies $\|\rho_S - \Delta(\rho_S)\|_1 \le 2\varepsilon_{eb}$ (\ref{app:channel}). This implication would fail for the unrestricted entanglement-breaking set: an entanglement-breaking channel may re-prepare coherent states, so separability of the output with respect to a reference does not by itself entail diagonality in the pointer basis; the restriction to $\mathcal{C}_X$ repairs this. We retain (A) explicitly for two reasons. First, it is operationally more accessible and can be verified through local measurements on $S$ alone. Second, (A) is a statement about the one state actually prepared, whereas (C) is a worst-case statement over all inputs; consequently (A) is certified no later than (C), and for generic initial states strictly earlier, for any admissible tolerances $\varepsilon_{eb} \le \varepsilon_d$ (Proposition~\ref{prop:AC} of Section~\ref{sec:toymodel}). The distinction between (A) and (C) therefore captures the temporal separation between decoherence and event anchoring: within this window, the system appears classical but is not yet certifiably irrecoverable.

\subsection{Freezing Time}

The freezing time is defined as
\begin{equation}
t_* = \inf \{ t \mid \text{Conditions (A), (B), and (C) are satisfied} \}.
\end{equation}
At $t = t_*$, the outcome becomes an anchored event in the sense of Definition~\ref{def:anchored}, providing a discrete and irreversible anchor for spacetime history.

\paragraph{Threshold Robustness}
The parameters $\varepsilon_d$, $\delta$, $R_*$, and $\varepsilon_{eb}$ define operational tolerances rather than fundamental constants. Physical conclusions are expected to be robust under reasonable variations of these parameters: for typical macroscopic interactions, all three conditions are satisfied rapidly and decisively, while for carefully isolated quantum systems, none are satisfied. The crossover regime, where threshold values become relevant, corresponds to mesoscopic situations where the boundary between quantum and classical behavior is itself physically meaningful. In this sense, the framework characterizes event anchoring as a crossover phenomenon rather than a sharp phase transition. The tolerances also admit a direct operational reading: two channels at diamond-norm distance $\varepsilon$ can be discriminated with single-use success probability at most $\tfrac{1}{2}(1+\varepsilon/2)$, so certifying residual recoverability at level $\varepsilon$ requires of order $1/\varepsilon^{2}$ independent interrogations \cite{Watrous2018}. The thresholds therefore encode the finite resolution of physically bounded observers---the same finite information-processing capacity discussed in Section~\ref{sec:bandwidth}---rather than arbitrary conventions.

% ======================================================================
\section{Finite Information-Processing Capacity}
\label{sec:bandwidth}
% ======================================================================

The emergence of local generative freezing relies on finite physical limits on information redistribution. Importantly, this limitation is not epistemic or cognitive, but physical: it reflects constraints on information flow imposed by locality, finite coupling strengths, and environmental degrees of freedom.

To characterize this capacity, we introduce a diagnostic quantity: the maximum rate at which system--environment correlations can develop per interaction step. In the explicit toy model of Section~\ref{sec:toymodel}, this corresponds to the per-step mutual information
\begin{equation}
I(X:E_k) = H_{\mathrm{bin}}\!\left(\frac{1+|\cos\theta|}{2}\right),
\label{eq:bandwidth}
\end{equation}
where $H_{\mathrm{bin}}$ denotes the binary entropy function, and which is directly determined by the coupling strength $\theta$ of the system--environment interaction. This quantity plays the role of a ``bandwidth'' $\Bandwidth$ for the information channel between system and environment.

In the toy model, coherence drops rapidly in the first few interaction steps---the regime where $\Bandwidth$ is most relevant---while later interactions primarily build classical redundancy (Condition B). The information-processing capacity therefore characterizes the early, dominant phase of information transfer from system to environment, rather than serving as a standalone general principle.

More generally, we characterize this capacity using quantum information-theoretic quantities \cite{NielsenChuang2000}. Let $\mathcal{N}_{S\to E}$ denote the quantum channel induced by the system--environment interaction. The ability of the system to maintain quantum coherence is bounded by the channel's capacity to transmit quantum information. When the effective quantum capacity satisfies
\begin{equation}
Q(\mathcal{N}_{S\to E}) \le \varepsilon_Q ,
\end{equation}
coherent phase information can no longer be preserved or recovered.

The relative entropy of coherence $C_{\mathrm{rel}}(\rho) = S(\Delta\rho) - S(\rho)$ \cite{Baumgratz2014} provides a natural measure of remaining quantum coherence in the system. In the toy model, $C_{\mathrm{rel}}$ drops to near zero within the first few collisions, confirming that the information-carrying capacity of the interaction channel is saturated early in the process.

In this sense, laboratory measurement may be viewed as an engineered regime in which the conditions for anchoring are realized in a controlled way, whereas natural events arise when analogous conditions are reached through uncontrolled environmental interactions.

% ======================================================================
\section{Explicit Toy Model: Repeated Collisional Interaction} % NEW
\label{sec:toymodel}
% ======================================================================

We now present a minimal analytic model that realizes the three freezing conditions within globally unitary dynamics. The model is based on repeated system--environment collisions, a standard framework for discrete open quantum dynamics \cite{Scarani2002,Ciccarello2022}.

\subsection{Model} % NEW

The system $S$ is a single qubit, initially prepared in the pure state
\begin{equation}
|\psi_S\rangle = \cos(\alpha/2)\,|0\rangle + \sin(\alpha/2)\,|1\rangle ,
\qquad \alpha \in (0,\pi),
\label{eq:initialstate}
\end{equation}
where $\alpha$ is the Bloch polar angle; the symmetric superposition $\alpha = \pi/2$, i.e.\ $|\psi_S\rangle = (|0\rangle + |1\rangle)/\sqrt{2}$, serves as the baseline case of Table~\ref{tab:freezing} and Figs.~\ref{fig:main}--\ref{fig:scan}. The environment consists of $N$ independent qubits $E_1, \dots, E_N$, each initially in $|0\rangle$. The system interacts sequentially with each environmental qubit through a controlled rotation:
\begin{equation}
U_{SE_k} = |0\rangle\langle 0|_S \otimes I_{E_k} + |1\rangle\langle 1|_S \otimes R_y(2\theta)_{E_k},
\end{equation}
where $R_y(2\theta)|0\rangle = \cos\theta|0\rangle + \sin\theta|1\rangle$ and $\theta \in (0, \pi/2)$ is the coupling strength. Each $E_k$ represents a distinct environmental degree of freedom (field mode, scattered photon, etc.). The controlled rotation correlates the system state with the environmental excitation, and each environmental qubit departs after interaction, modeling the physical inaccessibility of dispersed environmental degrees of freedom.

After $n$ collisions, the global state is
\begin{equation}
\begin{aligned}
|\Psi^{(n)}\rangle ={}& \cos(\alpha/2)\,|0\rangle|0\rangle^{\otimes n}
+ \sin(\alpha/2)\,|1\rangle|\varepsilon\rangle^{\otimes n}, \\
& |\varepsilon\rangle = \cos\theta|0\rangle + \sin\theta|1\rangle .
\end{aligned}
\label{eq:globalstate}
\end{equation}
This state remains pure at all times. The conditional environmental overlap $\langle 0|\varepsilon\rangle = \cos\theta$ controls all subsequent quantities.

This collisional model is an \emph{illustrative open-system realization} of the freezing conditions, not a derivation of them from relativistic quantum field theory, and the detector mapping described next is heuristic. While formulated here as a discrete qubit circuit, the model serves as a rigorously tractable stand-in for a localized particle detector (e.g., an Unruh--DeWitt detector) interacting sequentially with a relativistic quantum field \cite{Scarani2002,Ciccarello2022}. In this operational mapping, the system $S$ represents the internal degrees of freedom of a localized detector parameterized by its proper time $\tau$. Each environmental qubit $E_k$ corresponds to a localized wavepacket mode of the external field that interacts briefly with the detector and subsequently propagates away at the speed of light, modeling the causal inaccessibility of radiated field quanta. We employ this discrete formulation because evaluating channel-level irrecoverability---specifically, the diamond-norm distance to the nearest pointer measure-and-prepare channel (Condition~C)---in infinite-dimensional continuous QFT is mathematically highly intractable, whereas the collisional model admits exact analytical solutions that clearly demonstrate the separation of the anchoring conditions.

\subsection{Condition (A): Local Classicalization} % NEW

The reduced state of $S$ after $n$ collisions is
\begin{equation}
\rho_S^{(n)} = \begin{pmatrix} \cos^{2}(\alpha/2) & \tfrac{1}{2}\sin\alpha\,(\cos\theta)^{n} \\[3pt] \tfrac{1}{2}\sin\alpha\,(\cos\theta)^{n} & \sin^{2}(\alpha/2) \end{pmatrix}.
\end{equation}
The trace-norm distance to the dephased state is
\begin{equation}
\varepsilon_d(n;\alpha) = \big\|\rho_S^{(n)} - \Delta\big(\rho_S^{(n)}\big)\big\|_1 = |\sin\alpha|\,|\cos\theta|^{n},
\label{eq:condA}
\end{equation}
which depends on the prepared state through the coherence factor $|\sin\alpha| \le 1$, maximal for the symmetric input $\alpha = \pi/2$. Condition (A) is satisfied at
\begin{equation}
n_A = \big\lceil \ln\!\big(\varepsilon_d/|\sin\alpha|\big) \,/\, \ln|\cos\theta| \big\rceil
\end{equation}
when $|\sin\alpha| > \varepsilon_d$, and at $n_A = 0$ otherwise: an almost-classical preparation requires no decoherence at all.

\subsection{Condition (B): Redundant Record Formation} % NEW

Each environmental qubit carries the per-step mutual information $I(X:E_k)$ of Eq.~(\ref{eq:bandwidth}). For weak coupling ($\theta \ll 1$), a single fragment may carry insufficient information. We therefore consider groups of $m$ consecutive fragments $F_m = (E_{j+1}, \dots, E_{j+m})$, whose conditional states have overlap $(\cos\theta)^m$, giving
\begin{equation}
I(X:F_m) = H_{\mathrm{bin}}\!\left(\frac{1 + |\cos\theta|^m}{2}\right).
\label{eq:IXFm}
\end{equation}
The minimum group size $m_\delta$ is defined as the smallest $m$ such that $I(X:F_m) \geq (1-\delta)H(X)$, and the grouped redundancy is $R_\delta(n) = \lfloor n/m_\delta \rfloor$. Condition (B) is satisfied at $n_B = m_\delta \cdot R_*$.

The group-size threshold admits a closed form. Defining the overlap tolerance $\eta_\delta \equiv 2\,H_{\mathrm{bin}}^{-1}(1-\delta) - 1$, with the inverse taken on the branch $[\tfrac{1}{2}, 1]$, the requirement $I(X:F_m) \ge (1-\delta)H(X)$ is equivalent to $|\cos\theta|^m \le \eta_\delta$, so that
\begin{equation}
m_\delta = \left\lceil \frac{\ln(1/\eta_\delta)}{\ln(1/|\cos\theta|)} \right\rceil ,
\qquad
\eta_\delta \approx \sqrt{2\,\delta \ln 2} \quad (\delta \lesssim 0.2),
\label{eq:mdelta}
\end{equation}
the small-$\delta$ approximation holding to within a few percent. Equations~(\ref{eq:bandwidth}) and (\ref{eq:IXFm}) are evaluated for the symmetric input $\alpha = \pi/2$, for which $H(X) = 1$~bit; for general $\alpha$ the construction carries over with $H(X) = H_{\mathrm{bin}}\!\big(\cos^{2}(\alpha/2)\big)$ and the Holevo information of the two conditional pure states with priors $\cos^{2}(\alpha/2)$ and $\sin^{2}(\alpha/2)$, which likewise admits a closed form. Since the comparison between conditions (A) and (C) below does not involve condition (B), we retain the symmetric-input expressions here.

\subsection{Condition (C): Irreversibility Against Recovery} % NEW

The induced channel after $n$ collisions is a dephasing channel in the pointer basis with parameter $\lambda_n = |\cos\theta|^n$: populations are preserved, while coherences are multiplied by $\lambda_n$, irrespective of the input. Its diamond-norm distance to the set $\mathcal{C}_X$ of pointer measure-and-prepare channels is exactly $\lambda_n$ (\ref{app:channel}):
\begin{equation}
\varepsilon_{eb}(n) = \min_{\Phi \in \mathcal{C}_X} \|\mathcal{N}_n - \Phi\|_\diamond = |\cos\theta|^n,
\label{eq:condC}
\end{equation}
attained by the completely dephasing map $\Delta$. For the unrestricted entanglement-breaking set the distance is bracketed as $\lambda_n/2 \le \min_{\Phi \in \mathrm{EB}} \|\mathcal{N}_n - \Phi\|_\diamond \le \lambda_n$ (\ref{app:channel}), so none of the conclusions below depend on which set is used. Because (\ref{eq:condC}) is a property of the channel, it is independent of the initial state: condition (C) certifies that \emph{no} input---however coherent, and even when entangled with an external reference---can retain quantum correlations through the interaction. Condition (C) is satisfied at $n_C = \lceil \ln\varepsilon_{eb} / \ln|\cos\theta| \rceil$.

\paragraph{Remark: Two Notions of Separability}
Condition (C) concerns the entanglement-\emph{transmitting} capability of the reduced channel $\mathcal{N}_{S\to S}$. It implies approximate separability of the channel's Choi state---the joint state of the channel \emph{output} with an external reference that purifies the input---and within the model is implied by it up to a factor of two (\ref{app:channel}). This is logically independent of the separability of the system--environment state (\ref{eq:globalstate}); in fact, the two are anticorrelated. As $\theta \to \pi/2$, the conditional environment states become orthogonal, the global state (\ref{eq:globalstate}) approaches a maximally entangled (GHZ-type) state between $S$ and $E$, and the induced channel approaches the completely dephasing map $\Delta \in \mathcal{C}_X$: the environment then acts as a perfect pointer measurement. As $\theta \to 0$, the global state becomes a product state, no record forms, and the channel approaches the identity, maximally distant from $\mathcal{C}_X$.

The anticorrelation is enforced by the monogamy of entanglement: the more strongly $S$ entangles with $E$ (record formation), the less entanglement $S$ can retain with, or transmit to, any external reference (entanglement breaking). The same tradeoff is expressed quantitatively by the information--disturbance tradeoff and the continuity of the Stinespring dilation \cite{KSW2008}: a channel close to the identity has a complementary channel close to a constant map, depositing no record in $E$, and conversely. Record formation in the environment and irrecoverability of the system channel are thus two descriptions of one physical process, evaluated across two different bipartitions.

For the baseline presentation in Table~\ref{tab:freezing} and Figs.~\ref{fig:main}--\ref{fig:scan} we retain the illustrative tolerances $\varepsilon_d = 0.05$ and $\varepsilon_{eb} = 0.005$. The separation between conditions (A) and (C), however, does not rest on this choice. Comparing (\ref{eq:condA}) and (\ref{eq:condC}), up to integer rounding,
\begin{equation}
n_C - n_A \;\approx\; \frac{\ln(\varepsilon_d/\varepsilon_{eb}) + \ln(1/|\sin\alpha|)}{\ln(1/|\cos\theta|)} ,
\label{eq:gap}
\end{equation}
which is nonnegative whenever $\varepsilon_{eb} \le \varepsilon_d$ and remains strictly positive for generic initial states ($\alpha \ne \pi/2$) even at a \emph{common} threshold $\varepsilon_d = \varepsilon_{eb}$. The reason is structural rather than parametric: (A) certifies the coherence actually present in the prepared state, whereas (C) certifies the worst case over all inputs, and the two certifications coincide only for the maximally coherent, worst-case input $\alpha = \pi/2$. This is made precise in Proposition~\ref{prop:AC} and Fig.~\ref{fig:robust} below; within the window $n_A \le n < n_C$, the system appears classical but is not yet certifiably irrecoverable.

\subsection{Freezing Time, Unified Certification Formula, and Anchoring Regimes}

The freezing time is
\begin{equation}
t_* = n_* \cdot \tau, \quad n_* = \max(n_A,\, n_B,\, n_C),
\end{equation}
where $\tau$ is the collision interval. Substituting the closed forms for $n_A$, $n_B = m_\delta R_*$, and $n_C$ derived above, the freezing time is governed by a single rate--depth decomposition. Let
\begin{equation}
\Lambda \;\equiv\; \ln\frac{1}{|\cos\theta|} \;=\; -\ln\big|\langle 0|\varepsilon\rangle\big|
\end{equation}
denote the per-collision decay exponent of the conditional environmental overlap: one and the same microscopic rate drives all three certifications. Defining the certification depths
\begin{equation}
\mathcal{D}_A \equiv \ln\frac{|\sin\alpha|}{\varepsilon_d}, \qquad
\mathcal{D}_C \equiv \ln\frac{1}{\varepsilon_{eb}}, \qquad
d_B \equiv \ln\frac{1}{\eta_\delta},
\end{equation}
the freezing step takes the exact form
\begin{equation}
n_* \;=\; \max\!\Big( \big\lceil \mathcal{D}_A/\Lambda \big\rceil_{+} ,\;
\big\lceil \mathcal{D}_C/\Lambda \big\rceil ,\;
R_*\,\big\lceil d_B/\Lambda \big\rceil \Big),
\label{eq:master}
\end{equation}
where $\lceil x \rceil_{+} \equiv \max(0, \lceil x \rceil)$. Equation~(\ref{eq:master}) makes the role of every parameter transparent: $\varepsilon_d$ sets the depth demanded of \emph{local classical appearance}, $\varepsilon_{eb}$ the depth of \emph{quantum irrecoverability}, $R_*$ the demanded \emph{public accessibility}, $\delta$ (through $\eta_\delta$) the demanded \emph{per-record fidelity}, and the coupling $\theta$ sets only the clock rate $\Lambda$; the freezing step is the largest of the certification requirements. When $d_B/\Lambda \gtrsim 1$ (weak-to-intermediate coupling, $m_\delta \ge 2$), the ceilings may be dropped,
\begin{equation}
n_* \;\approx\; \frac{1}{\Lambda}\,\max\!\big( \mathcal{D}_A,\; \mathcal{D}_C,\; R_*\, d_B \big),
\label{eq:mastersmooth}
\end{equation}
whereas at strong coupling the redundancy term saturates its counting floor $n_B \ge R_*$---at least one collision is required per independent record---so that $n_* \to \max(1, R_*)$ as $\theta \to \pi/2$, and (\ref{eq:mastersmooth}) may underestimate $n_B$ by up to $R_* - 1$ collisions. In the opposite limit, $\Lambda \approx \theta^{2}/2$, so $t_* \approx 2\tau\,\max(\mathcal{D}_A, \mathcal{D}_C, R_* d_B)/\theta^{2}$ diverges as $\theta \to 0$: well-isolated systems never anchor.

Two structural consequences of (\ref{eq:master}) are threshold-independent.

\begin{prop}[Classicalization precedes irrecoverability]
\label{prop:AC}
For all $\theta \in (0, \pi/2)$, all initial states $\alpha \in (0,\pi)$, and all tolerances with $\varepsilon_{eb} \le \varepsilon_d$,
\begin{equation}
n_A \;\le\; n_C .
\end{equation}
\end{prop}

\begin{proof}
If $n$ certifies (C), i.e.\ $|\cos\theta|^{n} \le \varepsilon_{eb}$, then $|\sin\alpha|\,|\cos\theta|^{n} \le \varepsilon_{eb} \le \varepsilon_d$, so the same $n$ certifies (A). The set of steps certifying (C) is therefore contained in the set certifying (A), and their minima obey $n_A \le n_C$.
\end{proof}

\noindent
At a common threshold $\varepsilon_d = \varepsilon_{eb} = \varepsilon$, the gap (\ref{eq:gap}) reduces to $n_C - n_A \approx \ln(1/|\sin\alpha|)/\Lambda$, strictly positive for every $\alpha \ne \pi/2$ and growing without bound as the preparation approaches a pointer state ($n_A \to 0$ while $n_C$ is unchanged)---the cleanest illustration that (A) certifies a \emph{state} while (C) certifies a \emph{channel}. An immediate corollary is that condition (A) is \emph{never} the binding constraint: the anchoring time is always set by irrecoverability or by redundancy. This sharpens the claim that decoherence does not by itself constitute event formation: classical appearance is provably never the last certificate to be issued. Proposition~\ref{prop:AC} and the exactness of (\ref{eq:master}) were additionally confirmed by an independent numerical scan over $5 \times 10^{4}$ randomly sampled parameter sets $(\theta, \alpha, \varepsilon_d, \varepsilon_{eb}, R_*, \delta)$, with no violations.

Second, which of the two remaining conditions binds is governed by the dimensionless ratio
\begin{equation}
\chi \;\equiv\; \frac{R_*\, d_B}{\mathcal{D}_C} \;=\; \frac{R_*\, \ln(1/\eta_\delta)}{\ln(1/\varepsilon_{eb})} ,
\label{eq:chi}
\end{equation}
which contains neither $\theta$ nor $\alpha$: the coupling sets only the clock, while the tolerances set only the certification task. For $\chi < 1$, anchoring is \emph{irrecoverability-limited} ($n_C$ binds); for $\chi > 1$, it is \emph{redundancy-limited} ($n_B$ binds); for $|R_* d_B - \mathcal{D}_C| \lesssim R_* \Lambda$, the two certificates are issued within roughly one fragment-group step of each other and integer effects decide (a \emph{crossover band}, not a sharp transition). Which condition limits anchoring is thus a property of the certification task, not of the microscopic dynamics: different tolerance choices do reorder $n_B$ relative to $n_C$, and (\ref{eq:master}) specifies exactly how, while the dynamical content of the model---the exponential decay law and the rate $\Lambda$---is independent of these choices. The freezing time $t_* = \max(n_A, n_B, n_C)\,\tau$ is well defined in every regime precisely because its definition presupposes no ordering.

Table~\ref{tab:freezing} summarizes the baseline results (symmetric input $\alpha = \pi/2$) for representative coupling strengths.

\begin{table*}[t]
\centering
\caption{Freezing conditions for representative coupling strengths, with symmetric input $\alpha = \pi/2$ and $\varepsilon_d = 0.05$, $\varepsilon_{eb} = 0.005$, $R_* = 10$, $\delta = 0.2$ (giving $\chi \approx 1.3$: redundancy-limited regime).}
\label{tab:freezing}
\begin{tabular}{cccccccc}
\toprule
$\theta$ & $m_\delta$ & $I$/frag & $n_A$ & $n_C$ & $n_B$ & $n_*$ & A--C gap \\
\midrule
$15^\circ$ & 20 & 0.124 & 87 & 153 & 200 & 200 & 66 \\
$30^\circ$ & 5 & 0.355 & 21 & 37 & 50 & 50 & 16 \\
$45^\circ$ & 2 & 0.601 & 9 & 16 & 20 & 20 & 7 \\
$60^\circ$ & 1 & 0.811 & 5 & 8 & 10 & 10 & 3 \\
$75^\circ$ & 1 & 0.951 & 3 & 4 & 10 & 10 & 1 \\
\bottomrule
\end{tabular}
\end{table*}

Three coupling regimes are evident:
\begin{itemize}
    \item Weak coupling ($\theta \lesssim 30^\circ$): Grouped fragments are needed ($m_\delta \gg 1$), redundancy dominates ($n_B \gg n_C > n_A$), and freezing is slow. This corresponds to quantum systems that maintain coherence over many interaction steps.
    \item Intermediate coupling ($30^\circ \lesssim \theta \lesssim 60^\circ$): All three certificates are issued at well-separated steps ($n_A < n_C < n_B$ for these tolerances). This mesoscopic regime is where the distinction between decoherence and event anchoring is most operationally relevant.
    \item Strong coupling ($\theta \gtrsim 60^\circ$): Single fragments suffice ($m_\delta = 1$), and $n_*$ is limited mainly by the redundancy threshold $R_*$. This corresponds to macroscopic measurement-like interactions where anchoring occurs rapidly.
\end{itemize}

Figure~\ref{fig:main} illustrates the freezing dynamics for $\theta = \pi/3$. Panel~(a) shows the exponential decay $|\cos\theta|^n$ evaluated against the classicalization threshold $\varepsilon_d = 0.05$ and the stricter irrecoverability threshold $\varepsilon_{eb} = 0.005$, yielding the certification gap $n_C - n_A = 3$ for the symmetric input; the dashed curve shows the state-level decay $|\sin\alpha|\,|\cos\theta|^{n}$ for a generic preparation ($\alpha = \pi/6$), which is certified classical earlier still. Panel~(b) shows the linear growth of grouped redundancy $R_\delta(n)$, reaching $R_* = 10$ at $n_B = 10$. Panel~(c) displays the coherence drain and system entropy growth, confirming that coherence is effectively exhausted within the first few collisions while later interactions build classical redundancy. Panel~(d) presents the three-condition freezing diagram, showing the temporal ordering $n_A = 5 < n_C = 8 < n_B = 10 = n_*$.

\begin{figure*}[t]
\centering
\includegraphics[width=\textwidth]{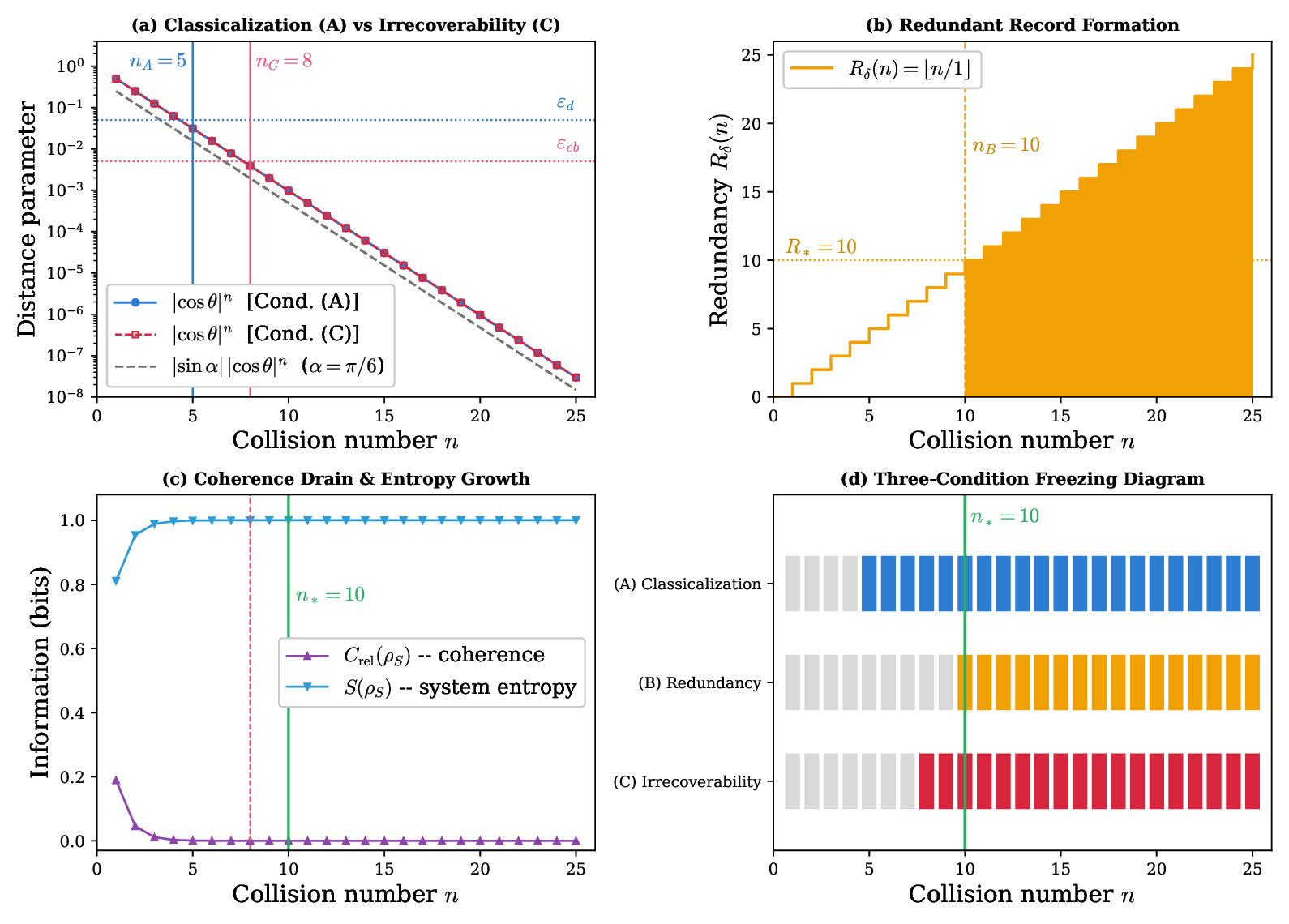}
\caption{Freezing dynamics for $\theta = 60^\circ$ ($|\cos\theta| = 0.5$), with $\varepsilon_d = 0.05$, $\varepsilon_{eb} = 0.005$, $R_* = 10$, $\delta = 0.2$. (a)~Classicalization (A) and irrecoverability (C) share the same exponential decay law but are evaluated against different thresholds; for the symmetric input $\alpha = \pi/2$ shown here (solid), the resulting gap of 3 collision steps is set by the ratio of the two illustrative tolerances. The gray dashed curve shows the state-level decay $|\sin\alpha|\,|\cos\theta|^{n}$ for a generic preparation ($\alpha = \pi/6$), for which condition (A) is certified strictly earlier than condition (C) even at a common threshold; the threshold-independent statement of the A--C gap is Proposition~\ref{prop:AC} and Fig.~\ref{fig:robust}(a,b). (b)~Grouped redundancy $R_\delta(n) = \lfloor n/m_\delta \rfloor$ grows linearly; for this coupling strength, single fragments suffice ($m_\delta = 1$). (c)~Relative entropy of coherence drops to near zero within the first few collisions, while system entropy saturates at 1~bit. (d)~Three-condition freezing diagram showing that classicalization (A) is satisfied first, irrecoverability (C) second, and redundancy (B) last, with the event anchored at $n_* = 10$.}
\label{fig:main}
\end{figure*}

Figure~\ref{fig:scan} shows the dependence of the freezing conditions on coupling strength $\theta$ across the full parameter range. Panel~(a) shows that for these tolerances ($\chi \approx 1.3 > 1$), $n_B$ dominates $n_*$ across all coupling strengths; Fig.~\ref{fig:robust}(c) delineates when redundancy is, and is not, the binding constraint. Panel~(b) shows the transition from the grouped-fragment regime ($m_\delta > 1$, weak coupling) to the single-fragment regime ($m_\delta = 1$, strong coupling). Panel~(c) displays the certification gap $n_C - n_A$, which is largest in the weak-coupling regime and decreases monotonically with increasing $\theta$, in accordance with the closed form $\ln(\varepsilon_d/\varepsilon_{eb})/\Lambda$ of (\ref{eq:gap}); the vertical axes of panels~(a) and (c) are clipped, as $n_B$, $n_*$, and the gap all diverge as $\theta \to 0$.

\begin{figure*}[t]
\centering
\includegraphics[width=\textwidth]{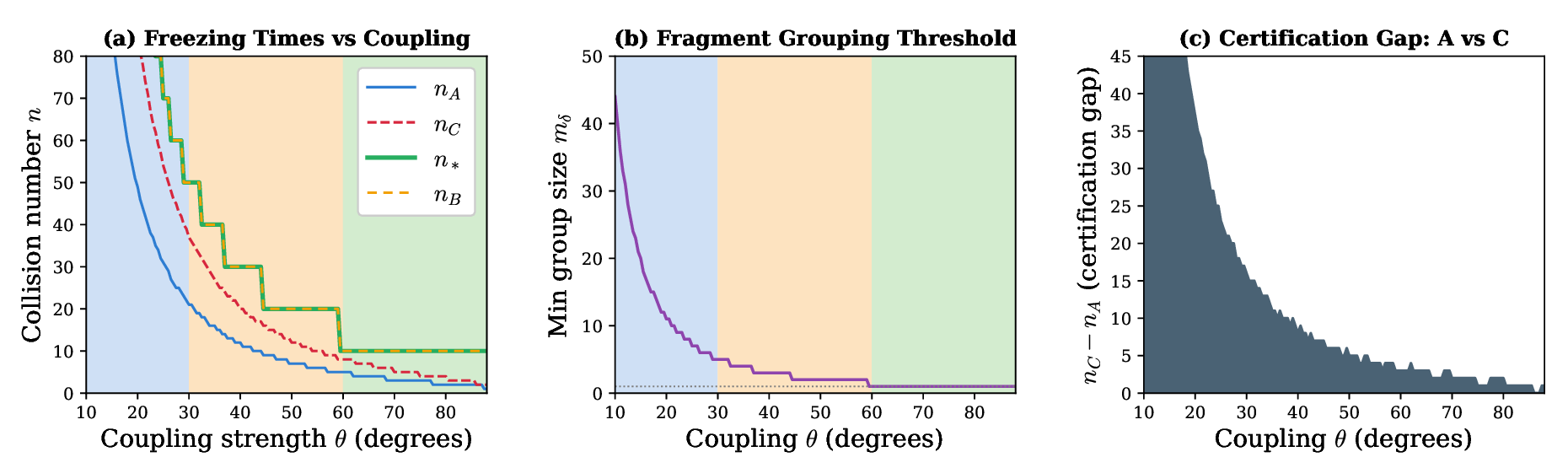}
\caption{Parameter scan across coupling strengths, with $\varepsilon_d = 0.05$, $\varepsilon_{eb} = 0.005$, $R_* = 10$, $\delta = 0.2$. (a)~Freezing times $n_A$, $n_C$, $n_B$, and $n_*$ as functions of $\theta$; for these tolerances $n_B = n_*$ throughout (dashed overlay). Three regimes are shaded: weak coupling (blue, $\theta \lesssim 30^\circ$), intermediate (orange, $30^\circ \lesssim \theta \lesssim 60^\circ$), and strong coupling (green, $\theta \gtrsim 60^\circ$). (b)~Minimum group size $m_\delta$ required for a fragment group to carry $(1-\delta)H(X)$ bits of information. (c)~Certification gap $n_C - n_A$ between state-level classicalization and channel-level irrecoverability, monotonically decreasing in $\theta$. The vertical axes of panels~(a) and (c) are clipped: all certification steps diverge as $\theta \to 0$ (well-isolated systems never anchor).}
\label{fig:scan}
\end{figure*}

Across Table~\ref{tab:freezing}, the ordering $n_A < n_C < n_B = n_*$ holds. Proposition~\ref{prop:AC} shows that the first inequality is universal and threshold-independent, whereas the position of $n_B$ relative to $n_C$ is genuinely regime-dependent: it is governed by the ratio $\chi$ of (\ref{eq:chi})---here $\chi \approx 1.3$, redundancy-limited---and encodes which physical resource limits anchoring, as mapped in Fig.~\ref{fig:robust}(c). The conclusion that survives every admissible parameter choice is therefore not a fixed total ordering but the following: the system looks classical no later than it becomes certifiably irrecoverable, and full anchoring awaits whichever channel-level certificate---irrecoverability or redundancy---is the more demanding. This confirms the paper's central claim that decoherence does not by itself constitute event formation.

Figure~\ref{fig:robust} establishes that this certification structure is a property of the model rather than of the tolerance choices; its three panels are three projections of the single formula (\ref{eq:master}). Panels~(a,b) evaluate the gap $n_C - n_A$ at a \emph{common} threshold $\varepsilon_d = \varepsilon_{eb} = 0.005$: the gap follows the closed form $\ln(1/|\sin\alpha|)/\Lambda$ of (\ref{eq:gap}), vanishes only at the worst-case input $\alpha = \pi/2$, and grows without bound for nearly classical preparations. Panel~(c) maps the binding constraint in the $(\delta, R_*)$ plane: the boundary between the irrecoverability-limited and redundancy-limited regimes follows the exact staircase $R_* = n_C/m_\delta$, tracked by the smooth curve $\chi = 1$ of (\ref{eq:chi}); by Proposition~\ref{prop:AC}, condition (A) never binds, so the diagram is independent of $\varepsilon_d$ for any $\varepsilon_d \ge \varepsilon_{eb}$.

\begin{figure*}[t]
\centering
\includegraphics[width=\textwidth]{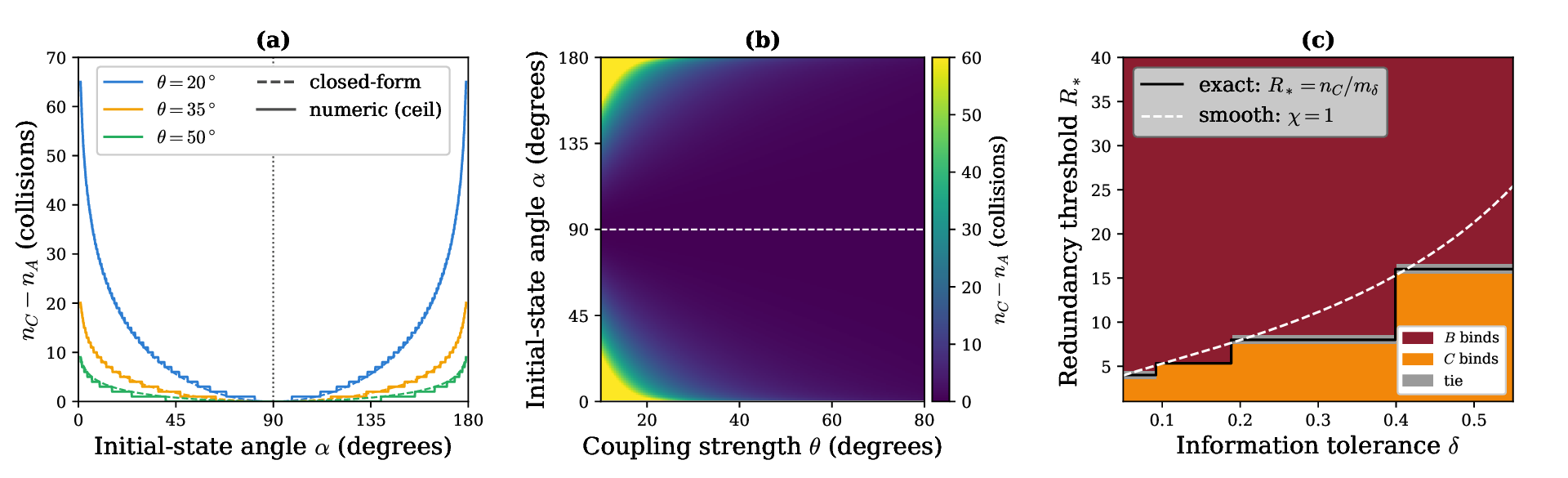}
\caption{Threshold robustness of the certification structure; all panels are projections of Eq.~(\ref{eq:master}). (a)~Certification gap $n_C - n_A$ at a common threshold $\varepsilon_d = \varepsilon_{eb} = 0.005$ as a function of the initial-state angle $\alpha$, for three coupling strengths. Solid curves: exact integer gaps; dashed curves: closed form $\ln(1/|\sin\alpha|)/\Lambda$. The unrounded closed-form gap vanishes only at the worst-case input $\alpha = \pi/2$ (dotted line) and diverges as the preparation approaches a pointer state ($n_A \to 0$ with $n_C$ unchanged). (b)~The same gap across the $(\theta, \alpha)$ plane; the white dashed line marks $\alpha = \pi/2$; the color scale is clipped at 60 collisions. (c)~Anchoring-regime diagram in the $(\delta, R_*)$ plane at $\theta = 45^\circ$, $\alpha = \pi/2$, $\varepsilon_{eb} = 0.005$: orange, irrecoverability-limited ($n_C$ binds); dark red, redundancy-limited ($n_B$ binds); gray, crossover (tie). Black staircase: exact boundary $R_* = n_C/m_\delta$; white dashed curve: smooth boundary $\chi = 1$ of Eq.~(\ref{eq:chi}). Condition (A) never binds (Proposition~\ref{prop:AC}), so the diagram is independent of $\varepsilon_d \ge \varepsilon_{eb}$; the regime boundaries are crossovers of width one fragment group, not sharp transitions.}
\label{fig:robust}
\end{figure*}

\subsection{Compatibility with Unitary Dynamics} % NEW

The global state $|\Psi^{(n)}\rangle$ remains pure at all times. All three conditions are properties of the reduced description:
\begin{itemize}
    \item (A) is a property of $\mathrm{Tr}_E(\rho_{SE})$,
    \item (B) is a property of correlations $I(X:E_k)$ within the global pure state,
    \item (C) is a property of the effective channel $\mathrm{Tr}_E \circ \mathcal{U}_{SE}$.
\end{itemize}
Irreversibility is effective and local: it describes the practical impossibility of an observer with access only to $S$ (or a small number of fragments) recovering pre-interaction coherence. The Poincar\'{e} recurrence time grows exponentially with $N$, making reversal operationally impossible for macroscopic environments. No violation of global unitarity is required or implied.

% ======================================================================
\section{Constructing Spacetime History}
\label{sec:constructing_history}
% ======================================================================

Having defined local generative freezing and established criteria for event anchoring, we now turn to the construction of spacetime history. The central claim of this framework is that, at the level of irreversibly anchored records, physical history is organized not as a continuous trajectory generated by quantum dynamics but as a partially ordered structure of anchored events---a statement about the operational organization of recorded history rather than about the fine structure of spacetime itself. This is not a claim that spacetime itself is fundamentally discrete.

\subsection{History as a Partially Ordered Set of Events}

We identify physical history $\History$ with a partially ordered set of events,
\begin{equation}
\History = \{\Event_a\},
\end{equation}
where the ordering relation $\Event_a \prec \Event_b$ holds if and only if $\Event_a$ can causally influence $\Event_b$ according to relativistic causality. This ordering is invariant under coordinate transformations and does not depend on any particular foliation of spacetime.

In this picture, the fundamental elements of history are not states or trajectories, but event-nodes. Once an event has been anchored---that is, once local generative freezing has occurred---it becomes a fixed element of the causal past of subsequent events. The irreversibility of history is thus attributed not to the breakdown of unitary dynamics, but to the irreversible embedding of events into the causal structure.

\subsection{Quantum Dynamics Between Events}

Between successive events, quantum systems evolve according to standard unitary dynamics. Quantum field theory governs the propagation of amplitudes, correlations, and interactions in the intervals separating event-nodes. These intervals need not be sharply bounded; in general, many microscopic interactions may occur without producing an event in the sense defined here.

The apparent continuity of spacetime and field evolution arises from the dense distribution of events in typical macroscopic environments. In regimes where event density is high, the causal skeleton formed by events can be approximated by a smooth spacetime manifold, and effective field theories provide accurate descriptions. In contrast, in regimes with sparse event formation---such as isolated quantum systems---the effective discreteness of recorded history becomes operationally relevant.

\subsection{Relation to Causal Set and Other Approaches} % REVISED TITLE

This event-centered conception of history bears structural resemblance to causal set theory \cite{Bombelli1987,Sorkin2003}, in that both emphasize a partially ordered structure of spacetime occurrences. However, the present framework differs in a crucial respect: events are not postulated as fundamental discreta, but are dynamically generated through quantum interactions subject to information-theoretic constraints. The framework thus offers a candidate micro-physical generation mechanism for the discrete nodes that causal set theory takes as axiomatic.

This distinction is significant. By deriving events from quantum interactions under finite information-processing capacity, the present approach avoids introducing discreteness as a fundamental axiom. Instead, an effective discreteness arises contextually, depending on interaction strength, environmental coupling, and information flow. The causal skeleton is not a pre-existing lattice but an emergent structure assembled through physical processes.

More broadly, the framework may be situated within a family of approaches that seek to understand spacetime structure in terms of quantum information. Holographic dualities and the geometry-from-entanglement program \cite{VanRaamsdonk2010} have demonstrated that spacetime connectivity can be related to entanglement structure. Loop quantum gravity and causal dynamical triangulations emphasize discrete or combinatorial structures underlying spacetime. The present proposal does not constitute a direct technical realization within any of these programs, but shares with them the view that aspects of effective spacetime structure may arise from underlying quantum processes. The specific contribution here is to propose operational criteria for when such processes produce irreversible spacetime anchors.

% ======================================================================
\section{Discussion}
\label{sec:discussion}
% ======================================================================
 
\subsection{Conceptual Scope: Event Anchoring versus Event Generation}
 
The framework developed in this work aims to clarify how irreversible spacetime history can emerge from fundamentally reversible quantum dynamics without modifying the unitary structure of quantum field theory. Its contribution is operational and structural rather than dynamical: it does not introduce new laws, but proposes information-theoretic criteria \cite{Wheeler1990} for when a quantum process yields a stable spacetime fact.
 
A distinction that we make explicit throughout the paper is between event generation and event anchoring. Event generation concerns how a particular outcome is first produced or selected from among alternatives, and with what probability---a question intimately connected to the Born rule, single-outcome selection, and the broader measurement problem. Event anchoring concerns a subsequent question: once a candidate outcome exists, when does it become a stable spacetime fact?
 
The present work concerns event anchoring, not a full theory of event generation. We regard anchoring as a necessary component of any eventual account of event generation, but we do not claim that it constitutes such an account on its own. In particular, the framework does not derive the Born rule, does not resolve single-outcome selection, and does not provide a full account of how a candidate outcome is first generated.
 
By event generation we mean the prior step in which a boundary interaction compresses a distributed relational structure between system and environment into a candidate local record, while leaving additional micro-information in inaccessible residual correlations. Formally, one may view this prior step schematically as a boundary compression
\begin{equation}
(\rho_{SE},\,\mathcal{I}_{\partial}) \;\longrightarrow\; (r,\,\rho_{\mathrm{res}}),
\end{equation}
where a joint system--environment state $\rho_{SE}$, mediated by the boundary information structure $\mathcal{I}_{\partial}$, yields a local record $r$ together with residual correlations $\rho_{\mathrm{res}}$ not retained in the event record. The present freezing criteria address the subsequent question: once such a candidate record $r$ exists, under what conditions does it become sufficiently classicalized, redundantly recorded, and effectively irrecoverable to function as a stable spacetime anchor?
 
In this sense, the framework reframes only one layer of the measurement problem: not how an outcome is selected, but how a candidate outcome, once present, becomes a stable spacetime fact. A fuller account of event generation remains open.

\subsection{Scope, Regimes, and Limitations}
 
The present work is primarily conceptual and foundational in nature. It does not predict deviations from standard quantum field theory in currently accessible experiments. The criteria for event anchoring are formulated in terms of quantum information-theoretic quantities that are, in principle, calculable but may be difficult to evaluate in practice for complex systems.
 
A mathematical remark is in order. Strictly speaking, perfect entanglement breaking ($\varepsilon_{eb} = 0$) might require infinite time or an infinite-dimensional environment. However, for all practical operational purposes, an exponentially small error $\varepsilon_{eb}$ is sufficient to define a physical event: once $\varepsilon_{eb}$ falls below any operationally meaningful threshold, no feasible physical process can distinguish the channel from a truly entanglement-breaking one, and the outcome is irreversibly anchored for all practical intents.
 
Although the framework does not predict deviations from standard QFT, the explicit toy model provides a concrete basis for identifying regimes where the distinction between decoherence and event anchoring becomes operationally relevant. For strong-coupling interactions, all three conditions are satisfied nearly simultaneously and anchoring occurs rapidly, consistent with the classical limit. For weaker-coupling mesoscopic interactions, the temporal separation between classicalization ($n_A$), irrecoverability ($n_C$), and full anchoring ($n_B$) becomes parametrically large, identifying a regime where the framework's distinctions could in principle be tested in platforms such as superconducting qubits or optomechanical systems. To illustrate: for strong-coupling macroscopic interactions ($\theta \sim 60^\circ$, $\tau \sim 10^{-13}\,s$ for photon scattering off a macroscopic object), the model gives $n_* = 10$ and hence $t_* \sim 10\tau \sim 10^{-12}\,s$, consistent with known decoherence timescales. For weak-coupling mesoscopic systems ($\theta \lesssim 20^\circ$), $n_*$ grows to $\sim 100$ or beyond, so that $t_*$ can be orders of magnitude longer, and the temporal gap between classicalization and full anchoring becomes parametrically large.
 
Regimes of extremely high event density---such as primordial decoherence during inflation or near black hole horizons---may amplify cumulative thermodynamic costs associated with information erasure and could, in principle, leave effective geometric signatures. We leave this possibility as an open question.

\subsection{Uniqueness of the Anchoring Criteria and Alternative Formulations}

The three conditions of Section~\ref{sec:freezing} are jointly sufficient operational certificates rather than a unique formalization of event anchoring, and it is natural to ask how the conclusions depend on the specific choices made. Conditions (A) and (B), taken together with the requirement that the conditional environment states associated with distinct outcomes become distinguishable, constitute a finite-tolerance counterpart of \emph{spectrum broadcast structure} \cite{Horodecki2015,Korbicz2021}, the structural benchmark for objectivity in quantum Darwinism: in the toy model, the conditional states of a fragment group $F_m$ have overlap $(\cos\theta)^m$, so the global state approaches spectrum broadcast structure precisely as the grouped records of Section~\ref{sec:toymodel} form.

Condition (C) likewise admits alternative formulations: one may demand instead that no recovery operation acting on $S$ restores the pre-interaction state above a fidelity threshold, with Fawzi--Renner-type bounds relating such recovery fidelities to conditional mutual information \cite{FawziRenner2015}; or one may measure the distance to the unrestricted entanglement-breaking set, which \ref{app:channel} shows to agree with our criterion to within a factor of two in the model.

Within the model, each of these alternative functionals is a monotone function of the single overlap parameter $\lambda_n = |\cos\theta|^n$; substituting them therefore amounts to a monotone reparameterization of the tolerances $(\varepsilon_d, \varepsilon_{eb}, \delta)$, which relocates numerical thresholds without altering the structure of Eq.~(\ref{eq:master}): the rate--depth factorization, Proposition~\ref{prop:AC}, and the regime classification of Fig.~\ref{fig:robust}(c) are insensitive to the particular certificates chosen. What is essential to the framework is the tripartite structure---a state-level, a record-level, and a channel-level requirement---rather than the specific functionals used to certify each.

\subsection{Localization and Covariant Extension}
 
The three anchoring conditions are formulated in terms of basis-independent information-theoretic quantities: trace-norm distance, mutual information, and diamond-norm channel distance. These are properties of states and channels, not of coordinate choices. However, a full covariant construction---specifying how the system--environment partition and pointer basis are defined across spacelike surfaces in a relativistic setting---remains an important open direction.
 
A concrete covariant recipe is suggested by the detector mapping of Section~\ref{sec:toymodel}. The system $S$ is a localized probe with worldline $\gamma$; the pointer basis is selected by the interaction Hamiltonian, in line with the commutant criterion of einselection \cite{Zurek2003}; and the environment fragments are the outgoing field modes that cross the successive light cones emanating from the coupling region, thereafter lying in the causal complement of the probe. The three conditions are then evaluated along $\gamma$ as functions of proper time, and the freezing time $t_*$ becomes a worldline-local, proper-time quantity. Different observers may describe the certification process in different coordinates, but anchored events enter history only through their causal relations, and the partial order of Section~\ref{sec:constructing_history} is frame-invariant.

Two structural obstacles deserve explicit mention. First, in algebraic QFT the local algebras of sharply bounded regions are generically of type~III, so strict partial traces and reduced density matrices are unavailable at that level; the split property, however, provides interpolating type-I factors \cite{Haag1996}, and the Fewster--Verch framework \cite{Mandrysch2024} formulates probe couplings at this level of rigor.

Second, the identification of approximately independent environment fragments must respect the vacuum correlations implied by the Reeh--Schlieder theorem; we expect the wavepacket localization of emitted quanta, as in the heuristic mapping of Section~\ref{sec:toymodel}, to supply the appropriate finite-tolerance notion of fragment independence. The repeated-collision model is a minimal analytic example meant to demonstrate that the criteria are realizable, not the final relativistic formulation; recent work on quantum reference frames \cite{Vilasini2025} provides further tools that may be relevant to a full covariant construction.

% ======================================================================
\section{Conclusion}
\label{sec:conclusion}
% ======================================================================

We have argued that an underappreciated aspect of the tension between quantum field theory and general relativity involves their descriptive primitives. GR is naturally organized around spacetime events and causal ordering, while QFT does not by itself specify when a quantum process becomes a definite spacetime fact.

We have proposed that this gap can be bridged by treating events as emergent structures anchored through local quantum--environment interactions under finite information-processing capacity. By introducing the notion of local generative freezing, we provided operational, information-theoretic criteria under which quantum potentiality condenses into spacetime facticity. An explicit repeated-collision toy model demonstrates that all three anchoring conditions can be satisfied within globally unitary dynamics, with a transparent and threshold-robust certification structure: classicalization is certified no later than channel-level irrecoverability for any admissible tolerances---strictly earlier for generic initial states---while the completion of anchoring is set by whichever requirement, irrecoverability or redundancy, is the more demanding (Proposition~\ref{prop:AC} and the closed form (\ref{eq:master})).

Spacetime history is then constructed as a partially ordered set of such events, forming a frozen causal skeleton upon which effective quantum field descriptions operate. The apparent continuity of spacetime and classical causality emerges as a coarse-grained approximation valid in regimes of high event density, while the record-level structure remains effectively discrete and irreversible at the level of anchored events---an operational discreteness of the record, not a fundamental discreteness of spacetime.

The present work addresses event anchoring---the conditions under which a candidate outcome becomes a stable spacetime fact---while leaving the question of event generation (outcome selection, probability assignment) as an important open problem. By clarifying the operational status of events and their role in constructing spacetime history, this work offers a candidate interface between quantum field theory and relativistic spacetime, grounded in an explicit model and delimited in its scope.

\section*{Acknowledgments}
The author thanks the anonymous reviewers for their constructive comments,
which materially strengthened the central claims of this work.

\section*{Data availability}

All results reported here follow from the closed-form expressions given in the text. The scripts used to generate Figs.~\ref{fig:main}--\ref{fig:robust} and to perform the $5\times10^{4}$-sample numerical verification of Proposition~\ref{prop:AC} and Eq.~(\ref{eq:master}) are available from the author upon reasonable request.

\section*{Declaration of generative AI and AI-assisted technologies in the manuscript preparation process}

During the preparation of this manuscript, the author used generative AI tools (including large language models) to assist with language refinement, structural organization, and clarity of presentation. All scientific arguments, interpretations, and conclusions were developed, reviewed, and verified by the author. The author takes full responsibility for the content of the published article.

\appendix

\section{Channel-Distance Computations and the (C)\,$\Rightarrow$\,(A) Implication}
\label{app:channel}

Throughout this appendix, $\mathcal{N}_\lambda$ denotes the qubit dephasing channel in the pointer basis, acting as $[\mathcal{N}_\lambda(\rho)]_{kk} = \rho_{kk}$ and $[\mathcal{N}_\lambda(\rho)]_{01} = \lambda\,\rho_{01}$, with $\lambda = \lambda_n = |\cos\theta|^n$ in the model of Section~\ref{sec:toymodel}; $\Delta = \mathcal{N}_0$ is the completely dephasing map, and $\mathcal{C}_X$ is the set of pointer measure-and-prepare channels defined in Section~\ref{sec:freezing}.

\subsection{Distance to the pointer measure-and-prepare set}
\label{app:cx}

\emph{Upper bound.} Since $\Delta \in \mathcal{C}_X$, it suffices to compute $\|\mathcal{N}_\lambda - \Delta\|_\diamond$. For any reference system $R$ and any joint state with block decomposition $\rho_{SR} = \bigl(\begin{smallmatrix} A & B \\ B^\dagger & C \end{smallmatrix}\bigr)$ with respect to the pointer basis of $S$,
\begin{equation}
\big[(\mathcal{N}_\lambda - \Delta) \otimes \mathrm{id}_R\big](\rho_{SR})
= \lambda \begin{pmatrix} 0 & B \\ B^\dagger & 0 \end{pmatrix},
\end{equation}
whose trace norm is $2\lambda \|B\|_1$. Positivity of $\rho_{SR}$ gives $B = A^{1/2} K C^{1/2}$ for some contraction $K$, hence $\|B\|_1 \le \|A^{1/2}\|_2\,\|C^{1/2}\|_2 = \sqrt{\mathrm{tr}A \cdot \mathrm{tr}C} \le \tfrac{1}{2}$, with equality for a maximally entangled input. Therefore $\|\mathcal{N}_\lambda - \Delta\|_\diamond = \lambda$.

\emph{Lower bound.} Every $\Phi \in \mathcal{C}_X$ outputs pointer-diagonal states for every input. Taking the input $|+\rangle\langle+|$ with $|+\rangle = (|0\rangle + |1\rangle)/\sqrt{2}$, the output $\mathcal{N}_\lambda(|+\rangle\langle+|)$ has off-diagonal element $\lambda/2$, and its trace distance to any diagonal state $\sigma$ is $2\sqrt{a^{2} + \lambda^{2}/4} \ge \lambda$, where $a$ is the population mismatch. Hence $\|\mathcal{N}_\lambda - \Phi\|_\diamond \ge \lambda$ for every $\Phi \in \mathcal{C}_X$, and
\begin{equation}
\min_{\Phi \in \mathcal{C}_X} \|\mathcal{N}_\lambda - \Phi\|_\diamond = \lambda ,
\end{equation}
which is Eq.~(\ref{eq:condC}).

\subsection{Distance to the entanglement-breaking set}
\label{app:eb}

The Choi state of $\mathcal{N}_\lambda$ is $J_\lambda = \tfrac{1}{2}\big(|00\rangle\langle 00| + |11\rangle\langle 11|\big) + \tfrac{\lambda}{2}\big(|00\rangle\langle 11| + |11\rangle\langle 00|\big)$, whose partial transpose $J_\lambda^\Gamma$ has spectrum $\{\tfrac{1}{2}, \tfrac{1}{2}, \tfrac{\lambda}{2}, -\tfrac{\lambda}{2}\}$. A channel is entanglement breaking if and only if its Choi state is separable \cite{Horodecki2003}, hence PPT, so $J_\Phi^\Gamma \ge 0$ for every $\Phi \in \mathrm{EB}$. Let $P$ project onto the negative eigenvector of $J_\lambda^\Gamma$; since $J_\lambda^\Gamma - J_\Phi^\Gamma$ is traceless and $\|2P - \mathbb{I}\|_\infty = 1$,
\begin{equation}
\|J_\lambda^\Gamma - J_\Phi^\Gamma\|_1
\;\ge\; \big|\mathrm{tr}\big[(J_\lambda^\Gamma - J_\Phi^\Gamma)(2P - \mathbb{I})\big]\big|
\;=\; \lambda + 2\,\mathrm{tr}\big[J_\Phi^\Gamma P\big] \;\ge\; \lambda .
\end{equation}
Since partial transposition on a qubit factor increases the trace norm by at most a factor of two \cite{Watrous2018}, and the diamond norm dominates the Choi-state trace distance,
\begin{equation}
\|\mathcal{N}_\lambda - \Phi\|_\diamond \;\ge\; \|J_\lambda - J_\Phi\|_1 \;\ge\; \tfrac{1}{2}\,\|J_\lambda^\Gamma - J_\Phi^\Gamma\|_1 \;\ge\; \frac{\lambda}{2} .
\end{equation}
Together with the upper bound of \ref{app:cx} (note $\Delta \in \mathcal{C}_X \subset \mathrm{EB}$),
\begin{equation}
\frac{\lambda}{2} \;\le\; \min_{\Phi \in \mathrm{EB}} \|\mathcal{N}_\lambda - \Phi\|_\diamond \;\le\; \lambda .
\end{equation}

\subsection{Condition (C) implies condition (A)}
\label{app:CtoA}

Suppose $\|\mathcal{N} - \Phi\|_\diamond \le \varepsilon$ for some $\Phi \in \mathcal{C}_X$, and let $\rho_S = \mathcal{N}(\rho_0)$ for an arbitrary input $\rho_0$. The state $\sigma \equiv \Phi(\rho_0)$ is pointer-diagonal, so $\sigma = \Delta\sigma$, and by contractivity of the trace norm under the channel $\Delta$,
\begin{equation}
\|\rho_S - \Delta\rho_S\|_1
\;\le\; \|\rho_S - \sigma\|_1 + \|\Delta\sigma - \Delta\rho_S\|_1
\;\le\; 2\,\|\rho_S - \sigma\|_1 \;\le\; 2\varepsilon .
\end{equation}
Thus condition (C) at tolerance $\varepsilon_{eb}$ implies condition (A) at tolerance $2\varepsilon_{eb}$, for every input state. The converse fails: condition (A) constrains only the one state actually prepared.

\end{document}